\documentclass[acmsmall]{acmart} % REMOVED: anonymous, review to get a clean final version

\acmSubmissionID{123-A5B-42} 


\usepackage{graphicx} 
\usepackage{algorithm}
\usepackage{algpseudocode}
\usepackage{dirtytalk}
\usepackage{bbm} % Kept for \mathbbm{E}
\usepackage{amsthm} % Explicitly loading amsthm for safety with \newtheorem.
\usepackage{mathtools}

\usepackage{subcaption} % Required for subfigures
\usepackage[labelformat=simple]{subcaption} % For (a) and (b) style labels
\newtheorem{fact}{Fact} 

\newcommand{\EE}[1]{\mathbbm{E}\left[#1\right]}
\newcommand{\VAR}[1]{\mathrm{var}\left[#1\right]}
\newcommand{\COV}[1]{\mathrm{Cov}\left[#1\right]}

\begin{document}

\title{Using random spanning trees in survivable networks design.}

% --- MANDATORY AUTHOR INFORMATION ---
% UNCOMMENT and FILL IN for the final camera-ready version.
\author{Dominik Bojko}
\author{B{\l}a{\. z}ej Wr{\' o}bel}
\affiliation{
  \institution{Wroc{\l}aw University of Science and Technology}
  \country{Poland}
}
\email{dominik.bojko@pwr.edu.pl}
\email{blazej.wrobel@pwr.edu.pl}% Replace with actual email

% --- ACM MANDATORY ABSTRACT ENVIRONMENT ---
\begin{abstract}
    We investigate a process of joining $k$ random spanning trees on a fixed clique $K_n$.
    The joined trees may not be disjoint and multiple edges are replaced by one simple edge.
    This process produces a simple graph $G$ on $n$~vertices with an edge set, which is a union of edge sets of the joined trees.
    We study a random variable $S_{k}$ of the number of edges in the generated graph $G$.
    The exact formula is derived for the expected value of the random variable $S_{k}$.
    In addition, an upper bound on the concentration coefficient of the random variable $S_{k}$ is provided.
    We use results of our analysis to design an algorithm to generate $k$-edge connected graphs for arbitrarily large values of $k \geq 2$.
    The designed algorithm solves a particular case of the Survivable Network Design Problem, where the cost of each edge is $c_{e} = 1$ and the connectivity requirement for each pair of vertices $u, v \in V(G)$ is $k$.
    The proposed algorithm is within a factor strictly less than $2$ of the optimal value (i.e., the number of edges in the generated graph) and its running time is $O(kn\log{n})$.
\end{abstract}

\maketitle % Renders the title block with author info

\section{Introduction}
\label{section:Introduction}
Simple graphs are natural models for various networks, e.g. electrical, road, or optical networks connecting various cities or countries.
Depending on the structure of a given network, it can be more or less immune to link failures.
A given network is reliable if it can serve its purpose even after a failure of some links which can happen naturally (e.g. harsh weather conditions, maintenance work, etc.) or are an act of sabotage.
For example, one should be able to ride from city $A$ to city $B$ even if some roads in the road network are closed due to maintenance or accidents.
Thus, designing reliable networks is an important problem from a practical point of view cf.~\cite{Zhenrong2004, Luss2004, Guolung2005, Panigrahi2011}.
\par

If we represent a given network as a simple graph, then link failures may be modeled as a removal of some subsets of edges.
The edge connectivity of the graph is the minimum number of edges that must be removed to disconnect it.
Thus, edge connectivity can be used to model the reliability of a given network.
Formally, a connected simple graph is $k$-edge connected, if $k$ is the minimal number of edges of the graph that have to be removed in order to lose connectivity of the graph.
From this point of view, the number $k$ can be treated as a level of reliability, which specifies how much the given network is immune to potential link failures.
The detailed literature on edge and vertex connectivity of graphs can be found in~\cite{Diestel2017,Bondy1976}~.
In this paper, we do not consider vertex connectivity. \par

Our research is inspired by the methods of construction of $2$ and $3$-edge connected graphs from~\cite{Diestel2017}.
For instance, every $2$-edge connected graph on $n$ vertices can be constructed by a process of adding paths between any two vertices, starting from a cycle graph.
Similarly, every $3$-edge connected graph can be constructed from $K_{4}$ clique by adding new edge $e$ between: two existing vertices, two subdividing vertices on different edges of the constructed graph or between one vertex from $K_{4}$ and one subdividing vertex.
The aforementioned methods are applicable only for specific values of connectivity parameter $k$.
This raises the question of whether a general construction exists for $k$-edge-connected graphs, regardless of the value of $k$.
The answer is positive and such general construction methods were proposed in~\cite{Su2009,Habib1980,Zhang89}, however they involve subroutines for finding $k - 1$ element cliques in graphs or solving instances of graph isomorphism problem, which makes these constructions rather difficult to use in practice.
\par

In contrast to construction methods from~\cite{Habib1980, Su2009, Zhang89}, we propose a simpler approach in which we use $k$ random spanning trees to construct $k$-edge connected graphs.
More specifically, we analyze the process of joining $k$ random spanning trees on $n$ vertices, resulting in a simple graph.
We select trees uniformly at random from the set of all spanning trees over $n$ vertices (see e.g.~\cite{Broder1989, Aldous1990} for such the randomization methods).
The edge set of the result graph is the union of the edge sets of selected trees.
\par

This process is suitable for constructing $k$-edge connected graphs. A main idea is as follows: for each pair $u,v\in V$ and a spanning tree, there is exactly one unique path between $u$ and $v$.
If two edge-disjoint random spanning trees (on the same set of vertices) are joined, then there are two edge-disjoint paths between each pair of vertices and Global Version of Menger's Theorem (cf. Appendix~\ref{appendix:mengers_theorem_statement}) implies that the constructed graph is $2$-edge connected.
Similarly, if $k$ edge-disjoint spanning trees are joined, then the obtained graph will be $k$-edge connected.
The assumption of edge-disjointness of the spanning trees is essential. Indeed, if all the trees are exactly the same, then the joint graph is still $1$ edge-connected.
\par

In practice, we are also concerned with the costs associated with links in the network.
Therefore, our aim is to design networks that are not only reliable, but also affordable.
Similar problem was studied in~\cite{Gabow1998,Goemans1994} and is known as the Survivable Network Design Problem (SNDP).
In discussed problem, the graph $G = (V, E)$ is given with the cost assigned to each edge $e \in E$.
The goal is to find minimum cost connected subgraph $H \subseteq G$ in which for each pair of vertices $u, v \in V$ their connectivity requirement is satisfied.
It is easy to see, that algorithms which solve instances of SNDP problem can be used to design networks which are both reliable and affordable.
\par

The discussed problem can be formulated as a particular instance of SNDP.
In SNDP we have an undirected graph $G = (V, E)$, a non-negative cost $c_{e}$ for every edge $e$, and a non-negative connectivity requirement $r_{uv}$ for every pair of vertices $u, v \in V$.
The aim is to find a minimum-cost connected subgraph in which each pair of vertices is joined by at least $r_{uv}$ edge-disjoint paths.
It is not hard to see that the proposed construction approach solves SNDP for the particular case when $G = K_{n}$, $c_{e} = 1$ for every $e \in E(K_{n})$ and $r_{uv} = k$ for every $u, v \in V$.
\par

\section{Related work}
As noted earlier, constructions for general values of $k$ were proposed in, e.g.,~\cite{Habib1980, Su2009, Zhang89}.
The problem with the methods mentioned is that they either test graph isomorphisms~\cite{Habib1980,Zhang89} or they use subroutines that search for cliques and check edge connectivity of intermediate graphs~\cite{Su2009}.
This may make these constructions rather difficult to implement and use in practice.
\par  

The method of generating a $k$-edge-connected graph by joining $k$ spanning trees of a graph is a direct application of the tree packing problem.
In tree packing problem we aim to find the maximum number of edge-disjoint spanning trees that can exist in a given graph.
In~\cite{Gabow95} the tree packing problem is used to design an algorithm for finding
edge connectivity parameter of a graph.
The input graph may be directed or undirected. The time complexity of discussed algorithm is $O(\lambda m \log(n^2/m))$ for directed graphs and $O(m + \lambda^2 n \log(n/\lambda))$ for undirected graphs, where $n = |V(G)|$, $m = |E(G)|$ and $\lambda$ is the edge connectivity parameter of the graph $G$.
\par

Randomized methods were used to solve problems related to cuts and network design.
In~\cite{Karger96} randomized algorithm for finding minimum cut in weighted, undirected and connected graphs was presented. Its time complexity is $O(n^{2}\log^{3}(n))$.
The discussed approach is also capable of finding all minimum cuts in a given graph.
In subsequent paper~\cite{Karger2000} the time bound was improved to $O(m\log^{3}(n))$ and it was also proven that all minimum cuts can be found in $O(n^{2}\log(n))$.
\par

Beyond min-cuts, random sampling theorems have led to faster algorithms for various network design problems~\cite{Karger1994}.
The technique, when combined with randomized rounding, was also used to achieve a near-optimal $(1+o(1))$-approximation for the NP-complete minimum $k$-connected subgraph problem when $k \gg \log n$.
\par

In~\cite{Gabow1998} the approximation algorithm for SNDP was presented. Its approximation factor is $2k - 1$ for $k \geq 2$ and $2$ for $k = 1$, where $k$ is the maximum connectivity requirement between any pair of nodes in a given problem instance.
Later, the approximation factor was improved to $2\mathcal{H}(k)$, where $\mathcal{H}(k) \approx \ln(k)$ is the $k-th$ harmonic number~\cite{Goemans1994}.
In~\cite{Jain2001} the dependence on $k$ was removed and $2$-approximation algorithm was presented.
The solution involved linear relaxation of the integer programming formulation of this problem and then rounding of derived solutions.
\par

In this paper, we present an alternative construction of $k$-edge connected graphs on $n$ vertices.
As stated in Section~\ref{section:Introduction}, we use the process of joining $k$ random spanning trees of clique $K_{n}$.
Multiple occurrences of the same edge are treated as one occurrence, therefore the result is simple graph.
This method can be easily used in practice, because each operation can be performed in polynomial time.
For generating random spanning trees we use approaches proposed in ~\cite{Wilson1996, Broder1989, Aldous1990}.
Their complexity is $O(n\log(n))$ on clique $K_{n}$ therefore, the overall complexity of $k$ trees selection is $O(kn\log(n))$.
\par  

It can be shown that this method may be used to solve particular instances of SNDP problem on clique $K_{n}$ ($n \geq 3$), where connectivity requirement for each pair of nodes $u, v \in V$ is equal $k$ and weight of edge $e \in E(K_{n})$ is equal $1$ (this means that the number of edges is the objective function to minimize).
Since each spanning tree has $n - 1$ edges, the result graph has at most $k(n - 1)$ edges.
In~section~\ref{section:finish} it is shown that, this algorithm is $2$-approximation for these particular instances of SNDP problem.
Note, that in the proposed approach, the dependency on connectivity $k$ requirement is also removed as in~\cite{Jain2001}.
The constraints used in linear programming model from~\cite{Jain2001} are defined for each proper subset $S \subset V$ and that means the number of such constraints may be exponential.
Our approach does not involve any integer or linear programming formulation, so it bypasses this flaw.
\par

\section{Selected properties of random spanning trees}
\label{section:expected_degree_of_fixed_vertex}

We consider labeled graphs and trees on the vertex set $V = [n] :=\{1, 2, \ldots, n \}$, i.e. each vertex is distinctly labeled by one number from $\{1, 2, \ldots, n \}$.
In this paper, we consider only uniform distribution over the space of all spanning trees (later referred to as ST) of a clique $K_{n}$.
More specifically, when we say that a tree is random, then we mean that this tree was selected with uniform distribution.
We adopt the terminology from~\cite{Goyal2009} and throughout the publication we will refer to result simple graphs as $k$-splicers.
In addition we extend the terminology and we will refer to result multigraphs, when repeating edges are included, as multi $k$-splicers.
\par
A Pr\"ufer code is a sequence of length $n - 2$, consisting of numbers from $[n]$.
There exists a natural bijection between labeled trees on $[n]$ and Pr\"ufer sequences of length $n-2$.
Thus, by drawing sequences of length $n - 2$ from the set $[n]$ uniformly at random (each label from $[n]$ is selected separately and independently), we also obtain a uniform distribution over the set of all labeled spanning trees on $[n]$. In all of the results and facts in this section, we assume that $n \geq 2$ and $1 \leq d \leq n - 1$. We will refer to set $[n]$ and clique $K_{n}$ interchangeably. The degree of a vertex $v$ in a graph $G$ will be denoted by $d(v)$. The following property of Pr\"ufer codes will be useful in the proof of Theorem~\ref{lemma:TreesWithoutEdges}.
\begin{fact}
\label{fact:vertex_degree_in_tree}
    Let $v \in V(K_{n})$ and $T$ be a ST of $K_{n}$.
    The degree of $v$ in $T$ is one more than the number of occurrences of its label label in the Pr\"ufer code of $T$.
\end{fact}

From Fact~\ref{fact:vertex_degree_in_tree} we immediately get the following lemma, which will also be useful in the proof of Theorem~\ref{lemma:TreesWithoutEdges}.
\begin{lemma}
\label{lemma:vertex_degree_in_tree}
    Let $v$ be a fixed vertex from the set $[n]$. The number of STs in which $d(v) = d$ is,
    \[
        \binom{n-2}{d-1} \cdot (n-1)^{n-1-d}.
    \]
\end{lemma}

Let us denote by $E(G)$, the set of edges of a graph $G$. 
The indicator random variable $X_e^G$ will be defined as follows:
\begin{equation}
\label{equation:indicator}
X^{G}_{e} = \begin{cases}
    1, & \textit{if $e \in E(G)$.} \\
    0, & \textit{otherwise.}
\end{cases}
\end{equation}

Using this notation, we will answer the question about the probability that a given edge $e = \{u, v\}$ is an edge in a random spanning tree $T$ chosen with uniform distribution.
\begin{lemma} Let $u, v \in [n]$ be different fixed vertices and $T$ be a random spanning tree on the set~$[n]$.
Let $X^{T}_{e}$ be the random indicator variable given by formula (\ref{equation:indicator}), where $e=\{u,v\}$.
%, such that $X^{T}_{e} = 1$ if $e = \{u, v\}$ is an edge in the tree $T$, and $X^{T}_{e} = 0$ otherwise.
Then
\[
\EE{X_{e}^T}= \Pr(X^{T}_{e}=1)=\frac{2}{n}~.
\]
\label{lemma:edge-probability}
\end{lemma}
\begin{proof}
Without a loss of generality, we may assume that $u=n-1$ and $v=n$.
Instead of ST $T$, we will consider its Pr\"{u}fer code. Assume that $e\in E(T)$. Note that the last label of the sequence is~either $n-1$ or $n$ (because $e$ connects these two nodes).
Other $n-3$ labels are arbitrary, so there are $2n^{n-3}$ such the Pr\"{u}fer codes.
Therefore,
$
\Pr(X^{T}_{e} = 1) = \frac{2n^{n - 3}}{n^{n - 2}} = \frac{2}{n}~.
$
\end{proof}

Now, we will prove a more general result regarding existence of two chosen edges $e$ and $w$ in RST. This result will be useful in the proof of Lemma~\ref{lemma:covariance_of_m}.

\begin{theorem} \label{lemma:TreesWithoutEdges}
Consider a random spanning tree $T$ on $n$ vertices and with edge set $E(T)$ and let $e$ and $w$ be any two different edges in $K_n$. It follows that
\begin{itemize}
\item $\Pr(X^{T}_{e} \cdot X^{T}_{w} = 1) = \frac{4}{n^{2}}$, if $e$ and $w$ are not adjacent.
\item $\Pr(X^{T}_{e} \cdot X^{T}_{w} = 1) = \frac{3}{n^{2}}$, if $e$ and $w$ are adjacent.
\end{itemize}
\end{theorem}

\begin{figure}[ht!]
    \centering
    \begin{subfigure}[t]{0.45\textwidth}
        \centering
        \includegraphics[width=0.5\linewidth, keepaspectratio]{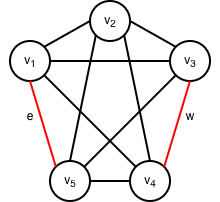} 
        \caption{}
        \label{fig:graph1-non-adjacent}
    \end{subfigure}
    \hfill
    \begin{subfigure}[t]{0.45\textwidth} 
        \centering
        \includegraphics[width=0.5\linewidth, keepaspectratio]{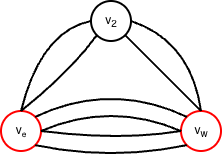} 
        \caption{}
        \label{fig:graph2-non-adjacent}
    \end{subfigure}
    
    \vspace{0.5cm}

    \begin{subfigure}[t]{0.45\textwidth} 
        \centering
        \includegraphics[width=0.5\linewidth, keepaspectratio]{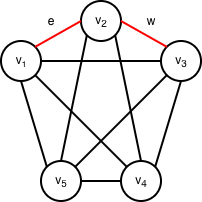} 
        \caption{}
        \label{fig:graph1-adjacent}
    \end{subfigure}
    \hfill
    \begin{subfigure}[t]{0.45\textwidth} 
        \centering
        \includegraphics[width=0.5\linewidth, keepaspectratio]{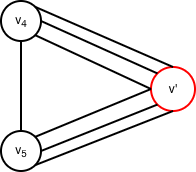} 
        \caption{}
        \label{fig:graph2-adjacent}
    \end{subfigure}
    
\caption{Contraction of edges in $K_{5}$. Top row: non-adjacent edges~\ref{fig:graph1-non-adjacent} before and~\ref{fig:graph2-non-adjacent} after contraction. Bottom row: adjacent edges~\ref{fig:graph1-adjacent} and~\ref{fig:graph2-adjacent} after contraction.}
    \label{fig:contraction-all}
\end{figure}

\begin{proof}
First, consider the case where two edges $e$ and $w$ are not adjacent. Let us contract the edges $e$ and $w$ --- we remove both edges from the complete graph $K_{n}$ and merge the vertices that were incident with them. We will call this operation a contraction of an edge. An example for $K_{5}$ is provided in figures~\ref{fig:graph1-non-adjacent} and~\ref{fig:graph2-non-adjacent}, where contracted edges are marked by red color and are not neighboring. After that we are left with the graph on $n - 2$ vertices, with two new vertices $v_{e}$ and~$v_{w}$. 
Let us derive the number of spanning trees of contracted graph. We will focus on those new and special vertices $v_{e}$ and $v_{w}$. Let $V$ be the set of vertices of a contracted graph. If $v \in V \setminus \{v_{e}, v_{w}\}$, then there are $2$ ways to connect $v$ with $v_{e}$ or $v_{w}$. We use Lemma~\ref{lemma:vertex_degree_in_tree} to find number of STs. There must be $i - 1$ occurrences of $v_{e}$ in Pr\"{u}fer sequence and $j -1$ occurrences of $v_{w}$ and the remaining $n - 2 - i - j$ positions can be occupied by any other vertices. 
Let us denote by $V$ the set of vertices of a graph with contracted edges. If $v \in V \setminus \{v_{e}, v_{w}\}$ then there are $2$ ways to connect $v$ to $v_{e}$ and $v_{w}$. If $v_{e}$ and~$v_{w}$ need to be connected, then there are $4$ ways to do that (cf. figure~\ref{fig:graph2-non-adjacent}), but it is already contained in the above formula. 
\begin{align*}
    \sum_{i=1}^{n-3} \sum_{j=1}^{n-2-i} \binom{n-4}{i-1,j-1,n-2-i-j} (n-4)^{n-2-i-j} 2^i 2^j =\\
    \sum_{i+j=0}^{n-4} \binom{n-4}{i,j,n-4-i-j} (n-4)^{n-4-i-j} 2^{i+j+2} = 4n^{n-4}~.
\end{align*}

Finally,
\[
    \Pr\left(X^{T}_{e} \cdot X^{T}_{w} = 1\right) = \frac{4}{n^{2}}~.
\]

Now we prove the case where the edges $e$ and $w$ are incident with some common vertex $v \in V\left(K_{n}\right)$. 
We again contract the edges $e$ and $w$ and as a result we obtain a graph on the $n - 2$ vertices with only one new vertex $v'$. An example for $K_{5}$ is provided in picture~\ref{fig:graph1-adjacent}, where contracted edges are marked by red color and are not neighboring. In this situation, we can extend the graph to $n$ vertices by joining each edge incident to $v'$ to one of the three endpoints -- the vertices which were endpoints of the edges $e$ and $w$ (cf.~\ref{fig:graph2-adjacent}). It can be done in $3^{d\left(v'\right)}$ ways, where $d(v')$ is the degree of the vertex $v'$ in a graph with contracted edges. From Lemma~\ref{lemma:vertex_degree_in_tree} and binomial theorem we obtain, 
\[
    \sum_{i = 0}^{n - 4} \binom{n - 4}{i} 3^{i + 1} (n - 3)^{n - 4 - i} = 3n^{n - 4}~.
\]

Similarly, as in the previous case of non-adjacent edges, there are $i$ occurrences of $v'$ label in the Pr\"ufer code and the rest of $n - 4 - i$ positions can be allocated for the remaining $n - 3$ labels. 
From the above, we obtain the thesis of the lemma when the edges $e$ and $w$ are adjacent.
\[
    \Pr\left(X^{T}_{e} \cdot X^{T}_{w} = 1 \right) = \frac{3}{n^{2}}~.
\]

This concludes the proof.
\end{proof}

\section{Expectation of the number of common edges in $k$ uniformly selected STs}

In this section, we provide a formula for the expected value of the number of common edges in $k$ independent random spanning trees i.e. the number of edges that are present in each tree $T_{i}$ for $i \in [k]$.
Let $T_{1}, T_{2}, \ldots, T_{k}$ denote random spanning trees of $K_n$ chosen independently and uniformly at random.
Let us define the random variable $C_{L} = \left|\bigcap\limits_{i \in L} E(T_{i}) \right|$, which is a cardinality of the set of common edges in all $|L|$  trees,  where $L \subseteq [k]$.
Now we are ready to formulate the first result in this section:

\begin{theorem} \label{CommonEdgesE}
Let $C_{L}$ be a random variable equal to the number of common edges in $|L|$ independently and uniformly selected RSTs $T_{i}$, for $i \in L$, where $L \subseteq [k]$.
Then
\[
\EE{C_{L}} = \binom{n}{2} \left(\frac{2}{n}\right)^{|L|} = 2^{|L|} \cdot \frac{n(n - 1)}{2n^{|L|}}~.
\]
\end{theorem}
\begin{proof}

Fix set $L$.
Since each edge must be present in each tree from $L$, the following holds:
\[
C_{L} = \sum_{e \in E(K_{n})} \prod_{j \in L} X^{T_{j}}_{e} ~.
\]
From the independence of the indicator random variables for $j \in L$ and Lemma~\ref{lemma:edge-probability}, we obtain the thesis of the theorem:
\begin{align*}
\EE{C_{L}} &= \sum_{e \in E(K_{n})} \EE{\prod_{j \in L} X^{T_{j}}_{e}} = \sum_{e \in E(K_{n})} \prod_{j \in L}\EE{X^{T_{j}}_{e}}\\
&= \binom{n}{2} \cdot \left(\frac{2}{n}\right)^{|L|} = \frac{n(n-1)}{2} \cdot \frac{2^{|L|}}{n^{|L|}}~.
\end{align*}
\end{proof} \par

Let us introduce the random variable $S_k:=\left|\bigcup_{i = 1}^{k} E(T_{i})\right|$, which represents the number of edges in a $k$-splicer.
The following fact expresses the random variable $S_{k}$ using the principle of inclusion and exclusion in terms of the number of common edges in selected RSTs.
\begin{fact} \label{fact:NumberOfCommonEdgesRandomVariable}
Let $S_{k}$ be a random variable of the number of edges in a $k$-splicer.
Then
\begin{equation}
\label{equation:sk_formula}
     S_{k} = \sum_{a=1}^{k}(-1)^{a+1} \sum_{L\in [k]^a} \left|\bigcap_{i\in L} E(T_i)\right|~.
\end{equation}
\end{fact}

Theorem~\ref{CommonEdgesE} and Fact~\ref{fact:NumberOfCommonEdgesRandomVariable} implies the following result for the random variable $S_{k}$.
\begin{theorem} \label{theorem:NumberOfEdgesInGraph}
Let $S_{k}$ be a random variable of the number of edges in a graph obtained by joining $k$~random spanning trees on $n$ vertices.
Then
\[
\EE{S_{k}} = k(n-k) + O\left(\frac{1}{n}\right).
\]
\end{theorem}
\begin{proof}
Taking the expectation of both sides of formula~(\ref{equation:sk_formula}), we get
\[
\EE{S_{k}} = \sum_{i = 1}^{k} (-1)^{i + 1}\binom{k}{i} \EE{C_{[i]}}~.
\]
Upon substituting the formula provided by Theorem~\ref{CommonEdgesE}, we deduce the thesis:

\begin{align*}
\EE{S_{k}} &= -\frac{n(n - 1)}{2} \sum_{i = 1}^{k} \binom{k}{i} (-1)^{i}\left(\frac{2}{n}\right)^{i} \\ 
&= -\binom{n}{2}\left[\left(1-\frac{2}{n}\right)^k -1\right] \quad \text{(by Binomial Theorem, adding and subtracting the term for $i = 0$)} \\  
&= k(n - 1) - \frac{k(k - 1)(n - 1)}{n} + O\left(\frac{1}{n}\right) \quad \text{(by Taylor Expansion)} \nonumber\\
&= k(n - 1) - k(k - 1) + O\left(\frac{1}{n}\right) 
=k(n-k) + O\left(\frac{1}{n}\right)~. \nonumber
\end{align*}
\end{proof}

\section{Random variable of the number of repeating edges}
\label{section:rv_of_number_of_repeating_edges}
In this Section we will examine $M$ which is the random variable of the~number of repeating edges. First, we will concentrate on a single edge $e \in E(K_{n})$ and introduce the random variable $R_{e}$ of the number of repetitions of that edge. Only surplus occurrences of edge $e$ are counted in the number of repetitions of $e$.
If edge $e$ occurred only once, then the number of its repetitions is zero. The following lemma allows us to express RV $R_{e}$ in terms of indicators $X^{T_{i}}_{e}$, where $i \in [k]$.

\begin{fact}
\label{fact:number_of_repeating_edges}
    Let $e \in E(K_{n})$ be an edge in the clique $K_{n}$ and let $R_{e}$ be the random variable of the number of repetitions of the edge $e$ in a multi $k$-splicer $G$. Then
    \begin{equation}
    \label{equation:re_alternate_form}
        R_{e} = \sum_{i = 1}^{k} X^{T_{i}}_{e} - \max(X^{T_{1}}_{e}, X^{T_{2}}_{e}, \ldots, X^{T_{k}}_{e})
    \end{equation}
\end{fact} 

From the Fact~\ref{fact:number_of_repeating_edges} the expected value of the random variable $M$ can be derived.
\begin{theorem}
\label{theorem:expected_value_of_m_rv}
    Let $M$ be the random variable of the number of multiple edges in a multi $k$-splicer $G$. Then
    \[
        \EE{M} = k(k - 1) + \frac{k(k - 1)(2k - 1)}{3n} + O\left(\frac{1}{n^{2}}\right)
    \]
\end{theorem}
\begin{proof}
From equation~(\ref{equation:re_alternate_form}) and the linearity of expectation, the expected value of $R_{e}$ can be derived.
\[
    \EE{R_{e}} = \sum_{i = 1}^{k} \EE{X^{T_{i}}_{e}} - \EE{\max(X^{T_{1}}_{e}, X^{T_{2}}_{e}, \ldots, X^{T_{k}}_{e})}
\]

Note that the maximum of indicator random variables $X^{T_{i}}_{e}$ can be expressed as follows.
\[
    \max(X^{T_{1}}_{e}, X^{T_{2}}_{e}, \ldots, X^{T_{k}}_{e}) = 1 - \prod_{i = 1}^{k} (1 - X^{T_{i}}_{e})~.
\]

Thus from Lemma~\ref{lemma:edge-probability}
\[
    \EE{R_{e}} = \frac{2k}{n} - 1 + \left(1 - \frac{2}{n}\right)^{k}~.
\]

The random variable $M$ can be expressed as
\[
    M = \sum_{e \in E(K_{n})} R_{e}~.
\]

From above, we obtain
\[
    \EE{M} = \sum_{e \in E(K_{n})} \EE{R_{e}} = \frac{n(n - 1)}{2} \left[\frac{2k}{n} - 1 + \left(1 - \frac{2}{n}\right)^{k}\right]~.
\]

Thus, by Taylor expansion,
\[
    \EE{M} = k(k - 1) + \frac{k(k - 1)(2k - 1)}{3n} + O\left(\frac{1}{n^{2}}\right)~.
\]

\end{proof}

In the proof of Theorem~\ref{theorem:expected_value_of_m_rv} the random variable $M$ is expressed in terms of the variables $R_e$, where $e \in E(K_n)$. Consequently, deriving the variance of $M$ requires an expression for the variance of $R_e$.
\begin{lemma}
\label{lemma:variance_of_repetitions_of_single_edge}
    Let $R_{e}$ be the random variable of the number of repeated occurrences of the edge $e \in E(K_{n})$ in a multi $k$-splicer $G$. Then
    \[
        \VAR{R_{e}} = \frac{2k}{n}\left(1 - \frac{2}{n}\right) + \left(1 - \frac{2}{n}\right)^{k} \left(1 - \frac{4k}{n} \right)- \left(1 - \frac{2}{n}\right)^{2k}~. 
    \]
\end{lemma}
\begin{proof}
    First, we calculate the second moment of RV $R_{e}$. Let us denote by $s_{e} = \sum_{i = 1}^{k} X^{T_{i}}_{e}$ and $P_{e} = \prod_{i = 1}^{k} (1 - X^{T_{i}}_{e})$. Then
    \[
        \EE{R^{2}_{e}} = \EE{(s_{e} - 1 + P_{e})^{2}} = \EE{s^{2}_{e}} - 2\EE{s_{e}} + 2\EE{s_{e}P_{e}} - 2\EE{P_{e}} + 1 + \EE{P^{2}_{e}}~.
    \]

    From Lemma~\ref{lemma:edge-probability} it is easy to see that $\EE{s_{e}} = \frac{2k}{n}$ and $\EE{P_{e}} = \EE{P^{2}_{e}} = \left(1 - \frac{2}{n}\right)^{k}$. Additionally Lemma~\ref{lemma:edge-probability} allows us to calculate rest of the required quantities:
    \begin{align*}
     \EE{s^{2}_{e}} &= \frac{2k}{n} + \frac{4k(k - 1)}{n^{2}}~,\\
     \EE{s_{e}P_{e}} &= 0 ~.
    \end{align*}

    Thus
    \[
        \EE{R^{2}_{e}} = \left(\frac{2k}{n} + \frac{4k(k - 1)}{n^{2}}\right) - \frac{4k}{n} + 1 - \left(1 - \frac{2}{n}\right)^{k}~.
    \]

    From the proof of Theorem~\ref{theorem:expected_value_of_m_rv} we obtain
    \[
        \left(\EE{R_{e}}\right)^{2} = \frac{4k^{2}}{n^{2}} - \frac{4k}{n} + \frac{4k}{n} \left(1 - \frac{2}{n}\right)^{k} + 1 - 2\left(1 - \frac{2}{n}\right)^{k} + \left(1 - \frac{2}{n}\right)^{2k}
    \]

    From the definition of variance we obtain the thesis of the lemma:
    \[
        \VAR{R_{e}} = \frac{2k}{n}\left(1 - \frac{2}{n}\right) + \left(1 - \frac{2}{n}\right)^{k} \left(1 - \frac{4k}{n} \right)- \left(1 - \frac{2}{n}\right)^{2k}~.
    \]
\end{proof}

% \begin{corollary}
% \label{cor:re_variance_asymptotic}
% Let $R_{e}$ be the random variable of the number of repeated occurrences of the edge $e \in E(K_{n})$ in a multi $k$-splicer $G$. Then

% \[
%     \VAR{R_{e}} = \frac{2k(k - 1)}{n^{2}} + O\left(\frac{1}{n^{3}}\right).
% \]
% \end{corollary}

Now, we are in a position to derive the variance of RV $M$. If the random variables $R_{e}$ and $R_{e'}$ for $e, e' \in E(K_{n})$ were independent, then we could simply sum the variance values for each edge $e \in E(K_{n})$. Unfortunately, this is not true, and RVs $R_{e}$ and $R_{e'}$ are dependent. The following auxiliary lemma will be useful in deriving the upper bound for variance of the discussed RV.
\begin{lemma}
\label{lemma:covariance_of_m}
    Let $R_{e}$ and $R_{e'}$ be random variables of the number of repeated occurrences of edges $e, e' \in E(K_{n})$ in a multi $k$-splicer. Then
    \begin{itemize}
        \item $\COV{R_{e}, R_{e'}} = 0$, if $e$ and $e'$ are not neighbouring,
        \item $\COV{R_{e}, R_{e'}} = -\frac{k}{n^{2}} + \frac{2k}{n^{2}}\left(1 - \frac{2}{n}\right)^{k -  1} + \left(1 - \frac{4}{n} + \frac{3}{n^{2}}\right)^{k} - \left(1 - \frac{2}{n}\right)^{2k}$, otherwise.
    \end{itemize}
\end{lemma}
\begin{proof}
From the proof of Theorem~\ref{theorem:expected_value_of_m_rv} we obtain
\[
    \left(\EE{R_{e}}\right)^{2} = \frac{4k^{2}}{n^{2}} - \frac{4k}{n} + \frac{4k}{n} \left(1 - \frac{2}{n}\right)^{k} + 1 - 2\left(1 - \frac{2}{n}\right)^{k} + \left(1 - \frac{2}{n}\right)^{2k}~.
\]

Now, we will derive the closed form for $\EE{R_{e} \cdot R_{e'}}$. The form of the random variable $R_{e}$ from the 
Fact~\ref{fact:number_of_repeating_edges} will be used and we will again denote by $s_{e} = \sum_{i = 1}^{k} X^{T_{i}}_{e}$ and $P_{e} = \prod_{i = 1}^{k} \left(1 - X^{T_{i}}_{e}\right)$. Then
\[
    \EE{R_{e} \cdot R_{e'}} = \EE{(s_{e} - 1 + P_{e})(s_{e'} - 1 + P_{e'})}~.
\]

Let us expand the above formula and use the linearity of expectation
\[
    \EE{R_{e} \cdot R_{e'}} \stackrel{d}{=} \EE{s_{e} \cdot s_{e'}} - 2\EE{s_{e}} - 2\EE{P_{e}} + 2\EE{s_{e}P_{e'}} + 1 + \EE{P_{e}P_{e'}}~.
\]

Now, we will consider the case where the edges $e, e' \in E(K_{n})$ are not neighboring. Then from Lemma~\ref{lemma:edge-probability} and Theorem~\ref{lemma:TreesWithoutEdges},
\begin{align*}
    \EE{s_{e}s_{e'}} &= \frac{4k^{2}}{n^{2}}~, \\
    \EE{s_{e}P_{e'}} &= \frac{2k}{n}\left(1 - \frac{2}{n}\right)^{k}, \\
    \EE{P_{e}P_{e'}} &= \left(1 - \frac{2}{n}\right)^{2k}~. \\
\end{align*}

Since $\EE{s_{e}} = \frac{2k}{n}$ and $\EE{P_{e}} = \left(1 - \frac{2}{n}\right)^{k}$, we obtain
\[
    \EE{R_{e}R_{e'}} = \frac{4k^{2}}{n^{2}} - \frac{4k}{n} + \frac{4k}{n}\left(1 - \frac{2}{n}\right)^{k} - 2\left(1 - \frac{2}{n}\right)^{k} + 1 + \left(1 - \frac{2}{n}\right)^{2k}~.
\]

Note, that in the case of $e, e' \in E(K_{n})$ not neighboring, we have
\(
    \EE{R_{e}R_{e'}} = \EE{R_{e}}\EE{R_{e'}},
\)
thus
\(
    \COV{R_{e}, R_{e'}} = 0.
\)

Now, we will consider the case where $e, e'$ are incident with some vertex $v \in [n]$. Then from Lemma~\ref{lemma:edge-probability} and Theorem~\ref{lemma:TreesWithoutEdges}
\begin{align*}
    \EE{s_{e}s_{e'}} &= \frac{4k^{2}}{n^{2}} - \frac{k}{n^{2}}~, \\
    \EE{s_{e}P_{e'}} &= k \left(\frac{2}{n} - \frac{3}{n^{2}}\right)\left(1 - \frac{2}{n}\right)^{k - 1}~, \\
    \EE{P_{e}P_{e'}} &= \left(1 - \frac{4}{n} + \frac{3}{n^{2}}\right)^{k}~.
\end{align*}

After derivation of $\EE{R_{e}R_{e'}}$ in the case where the edges $e$ and $e'$ are incident with some vertex, we obtain
\[
    \COV{R_{e}, R_{e'}} = -\frac{k}{n^{2}} + \frac{2k}{n^{2}}\left(1 - \frac{2}{n}\right)^{k -  1} + \left(1 - \frac{4}{n} + \frac{3}{n^{2}}\right)^{k} - \left(1 - \frac{2}{n}\right)^{2k}~.
\]
\end{proof}

\begin{corollary}
\label{corollary:variance_of_m_random_variable}
    Let $M$ be a random variable of the number of edge repetitions in $k$-splicer $G$. Then
    \[
        \VAR{M}= k(k - 1) - \frac{20k^{3} - 33k^{2} + 13k}{6n} + O\left(\frac{1}{n^{2}}\right)~.
    \]
\end{corollary}
\begin{proof}
We will use formula for the variance of sum of random variables. From the proof of Theorem~\ref{theorem:expected_value_of_m_rv} it follows that:

\begin{equation}
    \label{equation:variance_of_m}
    \VAR{M} = \VAR{\sum_{e \in E(K_{n})} R_{e}} = \sum_{e \in E(K_{n})} \VAR{R_{e}} + 2 \cdot \sum_{\mathclap{\substack{e \neq e' \\ e,e' \in E(K_{n})}}} \COV{R_{e}, R_{e'}}~.
\end{equation}

From Lemma~\ref{lemma:variance_of_repetitions_of_single_edge}, the first part of RHS of equation~(\ref{equation:variance_of_m}) can be written as

\begin{align*}
    \sum_{e \in E(K_{n})} \VAR{R_{e}} &= \frac{n(n - 1)}{2} \left[ \frac{2k}{n}\left(1 - \frac{2}{n}\right)  + \left(1 - \frac{2}{n}\right)^{k} \left(1 - \frac{4k}{n}\right) - \left(1 - \frac{2}{n}\right)^{2k}\right]~.
\end{align*}

By Lemma~\ref{lemma:covariance_of_m} it is known that covariance of $R_{e}$ and $R_{e'}$ is zero, when $e$ and $e'$ are not neighboring. Thus, it suffices to consider only pairs of neighboring edges $e$ and $e'$. The second part of RHS of equation~(\ref{equation:variance_of_m}) can be written as:

\begin{align*}
2 \cdot \sum_{\mathclap{\substack{e \neq e' \\ e,e' \in E(K_{n})}}} \COV{R_{e}, R_{e'}}  = n(n - 1)(n - 2) 
\left[ -\frac{k}{n^{2}} + \frac{2k}{n^{2}}\left(1 - \frac{2}{n}\right)^{k - 1} + \left(1 - \frac{4}{n} + \frac{3}{n^{2}}\right)^{k} - \left(1 - \frac{2}{n}\right)^{2k}\right]~.
\end{align*}

We will write asymptotic expansion for both parts of RHS of equation~(\ref{equation:variance_of_m}),

\begin{equation}
\label{equation:first_part_rhs}
\sum_{e \in E(K_{n})} \VAR{R_{e}} = k(k - 1) + \frac{2k^{3} - 9k^{2} + 7k}{3n} + O\left(\frac{1}{n^{2}}\right)~,
\end{equation}

\begin{equation}
\label{equation:second_part_rhs}
2 \cdot \sum_{\mathclap{\substack{e \neq e' \\ e,e' \in E(K_{n})}}} \COV{R_{e}, R_{e'}} = \frac{-8k^{3} + 17k^{2} - 9k}{2n} +O\left(\frac{1}{n^{2}}\right)~.
\end{equation}

Adding asymptotic expansions of~(\ref{equation:first_part_rhs}) and~(\ref{equation:second_part_rhs}), we obtain corollary.
\end{proof}

\begin{theorem}
\label{theorem:rv_m_concentration_coefficient}
    Let $M$ be the random variable of the number of repeating edges in a multi $k$-splicer~$G$. Then for $s > 0$,
    \[
        \lim_{n \rightarrow \infty} \Pr\left(|M - \EE{M}| \geq s\EE{M}\right) \leq \frac{1}{s^{2}k(k - 1)}~.
    \]
\end{theorem}
\begin{proof}
The proof follows from Corollary~\ref{corollary:variance_of_m_random_variable}, Theorem~\ref{theorem:expected_value_of_m_rv} and Chebyschev' inequality.
\end{proof}

\section{Algorithm for generating k-edge connected
graphs for arbitrarily chosen values of k}

We will now use generated random spanning trees to construct $k$-edge-connected graph. 
% First, the question of when the constructed graph is $k$-edge-connected will be answered. 
% The following theorem gives us the desired condition.
The following theorem provides us the condition, which determines $k$-edge-connectivity of the constructed graph.

\begin{theorem} \label{th:k-connected-condition}
If all of the joined random spanning trees $T_{1}, T_{2}, \ldots, T_{k}$ are pairwise edge disjoint, then a $k$-splicer is $k$-edge-connected.
\end{theorem}
\begin{proof}
Consider a pair of vertices $u, v \in V$. In each random spanning tree $T_{i}$ ($1 \leq i \leq k$) there is a unique path connecting the vertices $u$ and $v$. Since all the joined trees are disjoint, it follows that each tree $T_{i}$ adds one unique path between vertices $u$ and $v$. After joining $k$ spanning trees, there are $k$ edge-disjoint paths connecting the vertices $u$ and $v$. Thus, the result graph is $k$-edge-connected.
\end{proof}

\begin{algorithm}
\caption{Algorithm for generating $k$-edge-connected graphs.} \label{algorithm:k-edge-connected-generation}
\begin{algorithmic}[1]
\State $i \leftarrow 1$
\While{$i \leq k$} \label{algorithm:line:tree-generation-beg}
\State $T_{i} \leftarrow$ \text{random spanning tree chosen uniformly at random.}
\State $B_{T_{i}} \leftarrow G \cap T_{i}$ 
\State $G \leftarrow G \cup T_{i}$ \Comment{We need to know which edges are already in $G$.}
\State{$i \leftarrow i + 1$}
\EndWhile \\ \label{algorithm:line:tree-generation-end}

\For{$i \in \{2, 3, \ldots, k\}$} \label{algorithm:line:disjoint-phase-beg}
\For{$e \in B_{T_{i}}$}
\State $T_{i} \leftarrow T_{i} \setminus \{e\}$
\State $e' \leftarrow \text{GetReplacementEdge}(G, T_{i}, e)$ \label{algorithm:line:replacing-procedure}
\State $T_{i} \leftarrow T_{i} \cup \{e'\}$
\EndFor
\EndFor \label{algorithm:line:disjoint-phase-end}
\\

\State $G = \bigcup_{i = 1}^{k} T_{i}$ 
\State \Return $G$
\end{algorithmic}
\end{algorithm}

From Theorem~\ref{theorem:expected_value_of_m_rv} it is known that the average number of edge repetitions is $k(k - 1)$ and, from Theorem~\ref{theorem:rv_m_concentration_coefficient}, this number is well concentrated around its mean for $k \geq 3$. Thus, it is reasonable to replace each repeating edge in a tree, by the other edge, which does not occur in any of the selected trees. If we manage to replace each repeating edge with the new one, we will obtain $k$ edge-disjoint spanning trees, which then can be used to construct a $k$-edge-connected graph by summing sets of their edges. \par
The method for generating $k$-edge-connected graphs is presented in Algorithm~\ref{algorithm:k-edge-connected-generation}. In the beginning, $k$ random spanning trees of $K_{n}$ are generated. During generation process, each edge which is already present in partially generated graph $G$ is added to set $B_{T_{i}}$ which is the set of repeating edges in a tree $T_{i}$. Once $k$ random spanning trees have been generated, the phase of making them pairwise edge disjoint follows. For each tree $T_{i}$ where $i \geq 2$ (tree $T_{1}$ can be left intact), each repeating edge $e \in B_{T_{i}}$ is considered and the selection of replacement edge $e'$ follows (lines ~\ref{algorithm:line:disjoint-phase-beg}-\ref{algorithm:line:disjoint-phase-end} in Algorithm~\ref{algorithm:k-edge-connected-generation}). \par 
The selection process is presented in Algorithm~\ref{algorithm:find-replacement}. Let us assume that $e = \{u, v\}$ where $u, v \in K_{n}$. First, edge $e$ is removed from the current tree $T_{i}$, leaving only two connected components $C_{1}$ and~$C_{2}$. Then, vertex $u'$ from component $C_{1}$ with the smallest number of neighbors in component~$C_{2}$ is selected -- this vertex will be one of the endpoints of the replacement edge $e'$ (lines~\ref{algorithm:line:computing-neighbors-beg}-\ref{algorithm:line:computing-neighbors-end} in Algorithm~\ref{algorithm:find-replacement}). Next, possible second endpoints of edge $e'$ are determined by computing vertices from~$C_{2}$, which are not neighbors of endpoint $u' \in C_{1}$ (line~\ref{algorithm:line:computing-non-neighbor} in Algorithm~\ref{algorithm:find-replacement}). If the set of such non-neighboring vertices is not empty, then second endpoint $v'$ is chosen arbitrarily and replacement edge $e'= \{u', v'\}$ is returned. Otherwise, if no endpoint could be found in component $C_{2}$, then the algorithm returns the edge, which was to be removed. \par

\begin{algorithm}
\caption{Algorithm for finding replacement edges.}
\label{algorithm:find-replacement}
\begin{algorithmic}[1]
\Function{GetReplacementEdge}{G, T, e}
    \State $(C_1, C_2) \leftarrow \text{ConnectedComponents}(T)$ \\
    
    \State $\text{minDeg} \leftarrow \infty$
    \State $v^* \leftarrow \text{null}$
    \State $\text{neighbors} \leftarrow \emptyset$
    
    \Comment{Find vertex in $C_{1}$ with minimum number of neighbors in $C_{2}$.}
    \For{$v \in C_{1}$} \label{algorithm:line:computing-neighbors-beg}
        \State $N_v \leftarrow \{u \in C_{2} \mid \{v, u\} \in G\}$
        \If{$|N_v| < \text{minDeg}$}
            \State $\text{minDeg} \leftarrow |N_v|$
            \State $v^* \leftarrow v$
            \State $\text{neighbors} \leftarrow N_v$
        \EndIf
    \EndFor \label{algorithm:line:computing-neighbors-end} \\
    
    \State $\text{nonNeighbors} \leftarrow C_{2} \setminus \text{neighbors}$ \label{algorithm:line:computing-non-neighbor}
    
    \If{$\text{nonNeighbors} = \emptyset$}
        \Return $e$ \Comment{No replacement edge found - returning old edge.}
    \Else
        \State $w \leftarrow \text{arbitrarily chosen vertex from nonNeighbors.}$ 
        \State \Return $(v^*, w)$
    \EndIf
\EndFunction
\end{algorithmic}
\end{algorithm}

\begin{figure}[h!]
    \centering
    \begin{subfigure}[b]{0.27\textwidth}
        \centering
        \includegraphics[width=\textwidth]{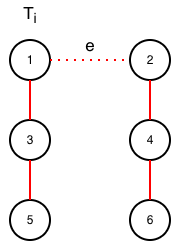}
        \caption{}
        \label{fig:generated_tree}
    \end{subfigure}
    \hfill
    \begin{subfigure}[b]{0.3\textwidth}
        \centering
        \includegraphics[width=\textwidth]{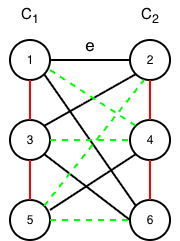}
        \caption{}
        \label{fig:possible_edges}
    \end{subfigure}
    \hfill
    \begin{subfigure}[b]{0.3\textwidth}
        \centering
        \includegraphics[width=\textwidth]{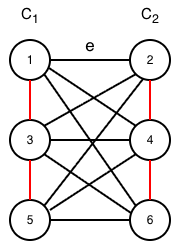} 
        \caption{}
        \label{fig:bipartite_complete}
    \end{subfigure}
    \caption{Finding a replacement for repeating edge $e$.}
    \label{fig:edge-replacement-process}
\end{figure}

The process of finding replacement edge is depicted in Figure~\ref{fig:edge-replacement-process}. Let us assume, that repeating edge $e$ from spanning tree $T_{i}$ $(2 \leq i \leq k)$ (Figure~\ref{fig:generated_tree}) is being processed. The possible new edges are presented in Figure~\ref{fig:possible_edges} with dashed, green lines. The vertex $5$ would be selected by Algorithm~\ref{algorithm:find-replacement} as the first endpoint of replacement edge $e'$. The second endpoint of $e'$ would be either vertices $2$ or $6$ -- in the case of vertices from component $C_{2}$, the choice is arbitrary. The case, when no new edge could be selected is shown in Figure~\ref{fig:bipartite_complete}. 
In this scenario, each vertex from component $C_{1}$ is connected with all vertices from component $C_{2}$, thus forming complete bipartite graph, and no new edge between these two components can be formed. Indeed, in such situation, we can simply return repeating edge $e$, because these two components are fully connected. \par

The time complexity of Algorithm~\ref{algorithm:k-edge-connected-generation} comprises two parts -- the part for generating $k$ random spanning trees of $K_{n}$ (lines~\ref{algorithm:line:tree-generation-beg}-\ref{algorithm:line:tree-generation-end}) and the part for replacing repeating edges with the new ones (line~\ref{algorithm:line:replacing-procedure}). The time complexity of the first part depends on the method of selecting RSTs with uniform distribution. Since the methods proposed in~\cite{Broder1989, Aldous1990, Wilson1996} are used in this paper, the time complexity of this part is $O(kn\log{n})$. It remains an open question whether the time complexity for the generation of random spanning trees can be improved to $O(kn)$, by generating Pr\"ufer sequences uniformly at random, as suggested in Section~\ref{section:expected_degree_of_fixed_vertex}. The replacement for each repeating edge can be found in $O(k(n - 1))$. Since, there are on average $k(k - 1)$ edge repetitions, the time complexity of this step is $O(k^{2}(k - 1)(n - 1))$. Thus, the overall time complexity of the algorithm is $O(kn\log{n} + k^{2}(k - 1)(n - 1))$. \par 
In our setting, $k$ is the parameter determined at a beginning and is constant. The higher value of $k$ is, the more edges we need in the network, so networks with large $k$, especially when $k$ is close to or equal $n$, may be expensive to construct. Therefore, we assume that $k \ll n $. On the other hand $n$ may be very large and it may change. \par

\section{Approximation Result}
\label{section:finish}
The designed algorithm solves instances of the Survivable Network Design Problem (SNDP) with cost $c_{e} = 1$ for each edge $e \in E$ and connectivity requirement $k$ for each pair of vertices $u, v \in V$ and returns the solution, which is within a factor strictly less than $2$ of optimum.

\begin{theorem}
The Algorithm~\ref{algorithm:k-edge-connected-generation} solves instances of the Survivable Network Design Problem with cost $1$ for each edge $e \in E$ and connectivity requirement $k$ for each pair of vertices $u, v \in V$ and returns the solution, which is within a factor strictly less than $2$ of optimum.
\end{theorem}
\begin{proof}
Let $OPT$ denote the number of edges in the optimal solution of the SNDP problem. Since the solution has to be a $k$-edge-connected graph then
\[
    OPT \geq \frac{kn}{2}~,
\]
where $n$ is the number of vertices. Let $S$ denote the number of edges in the solution found by the algorithm. The number of edges $S$ in a $k$-splicer is at most $k(n-1)$. Then
\[
    \frac{S}{OPT} \leq \frac{k(n - 1)}{\frac{kn}{2}} = \frac{2k(n - 1)}{kn} = \frac{2(n - 1)}{n} < 2~.
\]
\end{proof}

% --- ACM MANDATORY ACKNOWLEDGEMENTS ENVIRONMENT ---
\begin{acks}

\end{acks}

% --- ADDED: BIBLIOGRAPHY STYLING ---
\bibliographystyle{ACM-Reference-Format}
\bibliography{rst_bibliography} % Assuming your BibTeX file is named 'references.bib'

@inproceedings{Wilson1996,
author = {Wilson, David Bruce},
title = {Generating random spanning trees more quickly than the cover time},
year = {1996},
isbn = {0897917855},
publisher = {Association for Computing Machinery},
address = {New York, NY, USA},
url = {https://doi.org/10.1145/237814.237880},
doi = {10.1145/237814.237880},
booktitle = {Proceedings of the Twenty-Eighth Annual ACM Symposium on Theory of Computing},
pages = {296–303},
numpages = {8},
location = {Philadelphia, Pennsylvania, USA},
series = {STOC '96}
}

@INPROCEEDINGS{Broder1989,
  author={Broder, A.},
  booktitle={30th Annual Symposium on Foundations of Computer Science}, 
  title={Generating random spanning trees}, 
  year={1989},
  volume={},
  number={},
  pages={442-447},
  doi={10.1109/SFCS.1989.63516}
}

@Article{Gabow1998,
author={Gabow, Harold N.
and Goemans, Michel X.
and Williamson, David P.},
title={An efficient approximation algorithm for the survivable network design problem},
journal={Mathematical Programming},
year={1998},
month={Jun},
day={01},
volume={82},
number={1},
pages={13-40},
issn={1436-4646},
doi={10.1007/BF01585864},
url={https://doi.org/10.1007/BF01585864}
}

@article{Aldous1990,
author = {Aldous, David J.},
title = {The Random Walk Construction of Uniform Spanning Trees and Uniform Labelled Trees},
journal = {SIAM Journal on Discrete Mathematics},
volume = {3},
number = {4},
pages = {450-465},
year = {1990},
doi = {10.1137/0403039},
URL = { 
        https://doi.org/10.1137/0403039
},
eprint = { 
        https://doi.org/10.1137/0403039
},
}

@book{Diestel2017,
author = {Diestel, Reinhard},
title = {Graph Theory},
year = {2017},
isbn = {3662536218},
publisher = {Springer Publishing Company, Incorporated},
edition = {5th},
pages = {59--67}
}

@inproceedings{Goemans1994,
author = {Goemans, M. X. and Goldberg, A. V. and Plotkin, S. and Shmoys, D. B. and Tardos, \'{E}. and Williamson, D. P.},
title = {Improved approximation algorithms for network design problems},
year = {1994},
isbn = {0898713293},
publisher = {Society for Industrial and Applied Mathematics},
address = {USA},
booktitle = {Proceedings of the Fifth Annual ACM-SIAM Symposium on Discrete Algorithms},
pages = {223–232},
numpages = {10},
location = {Arlington, Virginia, USA},
series = {SODA '94}
}

@inproceedings{Karger1994,
author = {Karger, David R.},
title = {Random sampling in cut, flow, and network design problems},
year = {1994},
isbn = {0897916638},
publisher = {Association for Computing Machinery},
address = {New York, NY, USA},
url = {https://doi.org/10.1145/195058.195422},
doi = {10.1145/195058.195422},
booktitle = {Proceedings of the Twenty-Sixth Annual ACM Symposium on Theory of Computing},
pages = {648–657},
numpages = {10},
location = {Montreal, Quebec, Canada},
series = {STOC '94}
}

@article{Su2009,
title = {Removable edges in a k-connected graph and a construction method for k-connected graphs},
journal = {Discrete Mathematics},
volume = {309},
number = {10},
pages = {3161-3165},
year = {2009},
issn = {0012-365X},
doi = {https://doi.org/10.1016/j.disc.2008.09.005},
url = {https://www.sciencedirect.com/science/article/pii/S0012365X08005347},
author = {Jianji Su and Xiaofeng Guo and Liqiong Xu}
}

@incollection{Habib1980,
title = {A Construction Method For Minimally K-Edge-Connected Graphs},
editor = {Peter L. Hammer},
series = {Annals of Discrete Mathematics},
publisher = {Elsevier},
volume = {9},
pages = {199-204},
year = {1980},
booktitle = {Combinatorics 79},
issn = {0167-5060},
doi = {https://doi.org/10.1016/S0167-5060(08)70062-1},
url = {https://www.sciencedirect.com/science/article/pii/S0167506008700621},
author = {M. Habib and B. Peroche}
}

@book{Bondy1976,
  address = {New York},
  author = {Bondy, J. A. and Murty, U. S. R.},
  publisher = {Elsevier},
  title = {Graph Theory with Applications},
  year = 1976
}

@article{Zhenrong2004,
  author={Zhenrong Zhang and Wen-De Zhong and Mukherjee, B.},
  journal={IEEE Communications Letters}, 
  title={A heuristic method for design of survivable WDM networks with p-cycles}, 
  year={2004},
  volume={8},
  number={7},
  pages={467-469}
}

@inproceedings{Guolung2005,
author = {Guolong Zhu and QingJi Zeng and Tian Xu and Tong Ye and Junjie Yang},
title = {{Design of survivable WDM network using multigranularity p-cycles}},
volume = {5626},
booktitle = {Network Architectures, Management, and Applications II},
organization = {International Society for Optics and Photonics},
publisher = {SPIE},
pages = {1135 -- 1142},
year = {2005}
}

@article{Luss2004,
  author={Luss, H. and Wong, R.T.},
  journal={IEEE Transactions on Systems, Man, and Cybernetics - Part A: Systems and Humans}, 
  title={Survivable telecommunications network design under different types of failures}, 
  year={2004},
  volume={34},
  number={4},
  pages={521-530}
}

@inproceedings{Panigrahi2011,
author = {Panigrahi, Debmalya},
title = {Survivable network design problems in wireless networks},
year = {2011},
publisher = {Society for Industrial and Applied Mathematics},
address = {USA},
pages = {1014–1027},
numpages = {14},
location = {San Francisco, California},
series = {SODA '11}
}

@article{Zhang89,
author = {Zhang, Fuji and Guo, Xiaofeng and Chen, Rongsi},
title = {A Construction Method for Critically K-Edge-Connected Graphs},
journal = {Annals of the New York Academy of Sciences},
volume = {576},
number = {1},
pages = {663-670},
year = {1989}
}

@article{Gabow95,
title = {A Matroid Approach to Finding Edge Connectivity and Packing Arborescences},
journal = {Journal of Computer and System Sciences},
volume = {50},
number = {2},
pages = {259-273},
year = {1995},
issn = {0022-0000},
doi = {https://doi.org/10.1006/jcss.1995.1022},
url = {https://www.sciencedirect.com/science/article/pii/S0022000085710227},
author = {H.N. Gabow}
}

@article{Karger96,
author = {Karger, David R. and Stein, Clifford},
title = {A new approach to the minimum cut problem},
year = {1996},
issue_date = {July 1996},
publisher = {Association for Computing Machinery},
address = {New York, NY, USA},
volume = {43},
number = {4},
issn = {0004-5411},
url = {https://doi.org/10.1145/234533.234534},
doi = {10.1145/234533.234534},
journal = {J. ACM},
month = jul,
pages = {601–640},
numpages = {40}
}

@article{Karger2000,
author = {Karger, David R.},
title = {Minimum cuts in near-linear time},
year = {2000},
issue_date = {Jan. 2000},
publisher = {Association for Computing Machinery},
address = {New York, NY, USA},
volume = {47},
number = {1},
issn = {0004-5411},
url = {https://doi.org/10.1145/331605.331608},
doi = {10.1145/331605.331608},
journal = {J. ACM},
month = jan,
pages = {46–76},
numpages = {31}
}

@article{Jain2001,
author={Jain, Kamal},
title={A Factor 2 Approximation Algorithm for the Generalized Steiner Network Problem},
journal={Combinatorica},
year={2001},
month={Jan},
day={01},
volume={21},
number={1},
pages={39-60},
doi={10.1007/s004930170004},
url={https://doi.org/10.1007/s004930170004}
}

@inproceedings{Goyal2009,
author = {Goyal, Navin and Rademacher, Luis and Vempala, Santosh},
title = {Expanders via random spanning trees},
year = {2009},
publisher = {Society for Industrial and Applied Mathematics},
address = {USA},
pages = {576–585},
numpages = {10},
location = {New York, New York},
series = {SODA '09}
}

% --- BACK MATTER: APPENDICES ---
\appendix

\section{Menger's Theorem statement}
\label{appendix:mengers_theorem_statement}
In this section we provide a statement of Menger's Theorem from~\cite{Diestel2017}. Given graph $G = (V, E)$ and sets $A, B \subset V$, we call $P = x_{0}\ldots x_{k}$ an $A-B$ path if $V(P) \cap A = \{x_{0}\}$ and $V(P) \cap B = \{x_{k}\}$, where $x_{0} \in A$ and $x_{k} \in B$.
\par
\textbf{Menger's Theorem (Edge Form):} For any two distinct vertices $s$ and $t$ in a graph $G$, the maximum number of edge-disjoint $s-t$ paths is equal to the minimum number of edges whose deletion disconnects $s$ from $t$.

\end{document}